\documentclass[submission,copyright,creativecommons]{eptcs}
\providecommand{\event}{ICLP 2023} 

\usepackage{iftex}
\usepackage{amsmath}
\usepackage{amssymb}
\usepackage{amsthm}

\ifpdf
  \usepackage{underscore}         
  \usepackage[T1]{fontenc}        
\else
  \usepackage{breakurl}           
\fi

\newtheorem{example}{Example}
\newtheorem{theorem}{Theorem}
\newtheorem{remark}{Remark}
\newtheorem{corollary}{Corollary}

\DeclareMathOperator{\pa}{pa}

\DeclareMathOperator{\head}{head}
\DeclareMathOperator{\Do}{do}

\DeclareMathOperator{\body}{body}

\DeclareMathOperator{\effect}{effect}
\DeclareMathOperator{\causes}{causes}

\DeclareMathOperator{\graph}{Graph}
\DeclareMathOperator{\Graph}{Graph}

\DeclareMathOperator{\FCM}{FCM}
\DeclareMathOperator{\Ind}{Ind}

\DeclareMathOperator{\LP}{LP}
\DeclareMathOperator{\Facts}{Facts}

\title{``Would life be more interesting if I were in AI?'' \\ 
Answering Counterfactuals based on \\ Probabilistic Inductive Logic Programming}
\author{Kilian Rückschlo\ss
\institute{ Ludwig-Maximilians-Universität München \\
			Oettingenstra\ss{e} 67, 80538 München, Germany\\
         \email{kilian.rueckschloss@lmu.de}}
\and
Felix Weitkämper
\institute{ Ludwig-Maximilians-Universität München \\
			Oettingenstra\ss{e} 67, 80538 München, Germany\\
         \email{felix.weitkaemper@lmu.de}}
}    

\newcommand{\titlerunning}{Answering Probabilistic Counterfactuals based on ILP}
\newcommand{\authorrunning}{K.~Rückschlo\ss, F.~Weitkämper}

\hypersetup{
  bookmarksnumbered,
  pdftitle    = {\titlerunning},
  pdfauthor   = {\authorrunning},
  pdfsubject  = {EPTCS},               
}

\begin{document}
\maketitle

\begin{abstract}
Probabilistic logic programs are logic programs where some facts hold with a specified probability. Here, we investigate these programs with a causal framework that allows counterfactual queries. Learning the program structure from observational data is usually done through heuristic search relying on statistical tests. However, these statistical tests lack information about the causal mechanism generating the data, which makes it unfeasible to use the resulting programs for counterfactual reasoning. To address this, we propose a language fragment that allows reconstructing a program from its induced distribution. This further enables us to learn programs supporting counterfactual queries.
\end{abstract}

\section{Introduction}
While only observing the world, humans are used to drawing counterfactual conclusions, i.e.~they reason about how events would have unfolded under different circumstances. This leads us to judgements like: ``I would have published more papers, if I were in AI.'' without actually experiencing the alternative reality in which we work in AI. Note that this capability allows us to make sense of the past, to plan courses of actions, to make emotional and social judgments as well as to adapt our behaviour \cite{CounterfactualIntroduction}.  Hence, in artificial intelligence one also wants to infer a model of the world that supports counterfactual reasoning.

Currently, the \textsc{WhatIf}-solver \cite{whatif} establishes a counterfactual reasoning for ProbLog programs \cite{Problog}, i.e.~logic programs in which each clause holds with a specified probability. However, is the counterfactual reasoning provided by a ProbLog program uniquely determined by its distribution semantics \cite{DistributionSemantics}?

Assume for instance that a patient is treated, denoted $treatment$, with a probability of $0.5$. If we treat a patient, we expect him to recover, denoted $recovery$, with a probability of $0.7$, otherwise he recovers with a probability of $0.5$. The resulting  distribution can be encoded with the following two programs~$\textbf{P}_{1/2}$.
\begin{align*}
& \textbf{P}_1:			   		  
&& 0.5 :: treatment 
&& 0.5 ::recovery  
&& 0.4 ::recovery \leftarrow treatment \\										    
&  \textbf{P}_2: 
&& 0.5::treatment
&& 0.5::recovery \leftarrow \neg treatment 
&& 0.7::recovery \leftarrow treatment
\end{align*}

Assume further that the patient recovers while he has not been treated. What is the probability that he would have recovered under treatment?

In both programs $\textbf{P}_{1/2}$, we conclude from our observations that the patient recovers because of the second clause, i.e.~we conclude that the second clause holds in the world we observe. If we had additionally treated the patient, under program $\textbf{P}_1$, he would still have recovered as the second clause in $\textbf{P}_1$ is still applicable under treatment. Hence, we obtain a probability of one for the patient to recover under treatment.  Whereas, in program $\textbf{P}_2$, the second clause is not applicable under treatment. Hence, in this case, if we had treated the patient here, he can only recover because of the third clause and we obtain a probability of $0.7$ for the patient to recover under treatment. 

As we see, in general, the classical distribution semantics \cite{DistributionSemantics} does not uniquely determine the outcome of a counterfactual query. Further, note that ProbLog programs are usually learned from observations sampled from a distribution of interest.  Hence, even if we assume perfect learning, we can only ensure to obtain a program representing the correct distribution. In particular,~a structure learning algorithm is not able to distinguish the programs $\textbf{P}_{1/2}$, i.e. we cannot ensure that a learned program answers counterfactual queries correctly.

In this contribution, we present a fragment in which each program is uniquely determined by its class dependency graph and the corresponding distribution. We further argue that this yields a setting for the available structure learning methods which supports counterfactual reasoning.      
  
\section{Foundations} \label{section - preliminaries}

Here, we introduce ProbLog programs \cite{Problog} and we recall how counterfactual queries are processed on them \cite{whatif}. Finally, we quickly explain the design of the currently available structure learning algorithms for these programs.

As the semantics of non-ground ProbLog programs is usually defined by grounding, we restrict ourselves to the propositional case. Hence, we construct our programs from a \textbf{propositional alphabet} that is given by a finite set of propositions $\mathfrak{P}$ together with a subset $\mathfrak{E}(\mathfrak{P}) \subseteq \mathfrak{P}$ of \textbf{external proposition}. In this context, we call $\mathfrak{I}(\mathfrak{P}) := \mathfrak{P} \setminus \mathfrak{E}(\mathfrak{P})$ the set of \textbf{internal propositions}.

A \textbf{literal} $l$ is an expression $p$ or $\neg p$ for a proposition $p \in \mathfrak{P}$. We call~$l$ a \textbf{positive} literal if it is of the form $p$ and a \textbf{negative} literal if it is of the form~$\neg p$. Further, we call $l$ an \textbf{external} or \textbf{internal} literal if~${p \in \mathfrak{E}(\mathfrak{P})}$ or $p \in \mathfrak{I}(\mathfrak{P})$ respectively.
A \textbf{(logical) clause} $LC$ is an expression $h \leftarrow b_1,...,b_n$ where~${h \in \mathfrak{I}(\mathfrak{P})}$ is an internal proposition called the \textbf{head} and where $\{ b_1,...,b_n \}$ is a finite set of literals called the \textbf{body} of $LC$.
Finally, a \textbf{random fact} $RF$ is an expression $\pi(RF) :: u(RF)$ where $\pi(RF) \in [0,1]$ is the \textbf{probability} and where $u(RF) \in \mathfrak{E} (\mathfrak{P})$ is an external proposition called the \textbf{proposition} of $RF$.

\begin{example}
Consider the alphabet~$\mathfrak{P} := \{ treatment,~recovery, u_1,u_2,u_3 \}$ with external literals $\mathfrak{E}(\mathfrak{P})$ given by $\{u_1,u_2,u_3\}$.
We have that $treatment$ is a positive literal, whereas $\neg treatment$ is a negative literal. Further, $recovery \leftarrow treatment, u_3$ is a clause and $0.4 :: u_3$ is a random fact.
\label{example - expressions}
\end{example}

Now a \textbf{logic program} is a finite set of clauses and a \textbf{ProbLog program} $\textbf{P}$ is given by a logic program~$\LP (\textbf{P})$ and a set~$\Facts (\textbf{P})$ consisting of a unique random fact for every external proposition. In this case, we call~$\LP (\textbf{P})$ the \textbf{underlying logic program} of $\textbf{P}$. Finally, we define ($\mathfrak{P}$-)formulas $\phi$ and ($\mathfrak{P}$-)structures ${\mathcal{M}: \mathfrak{P} \rightarrow \{ True, False \}}$ as usual in propositional logic. Whether a given $\mathfrak{P}$-structure ${\mathcal{M}}$ satisfies a formula $\phi$, written $\mathcal{M} \models \phi$, is also defined as usual in propositional logic.

\begin{example}
In the alphabet $\mathfrak{P}$ of Example \ref{example - expressions} we can write the following ProbLog program.
\begin{align*}
& 0.5 :: u_1 && 0.5 :: u_2 && 0.4 :: u_3 
&& treatment \leftarrow u_1 && recovery \leftarrow u_2 && recovery \leftarrow treatment, u_3
\end{align*}
\label{example - program}
\end{example}

As the semantics of ProbLog programs we choose the FCM-semantics \cite{fcm-semantics}, which supports counterfactual reasoning:
For a ProbLog program $\textbf{P}$ we define the \textbf{functional causal models semantics} or~\mbox{\textbf{FCM-semantics}} to be the system of Boolean equations
$$
\FCM (\textbf{P}) :=
\left\{ 
p^{\FCM} := \bigvee_{\substack{LC \in \LP (\textbf{P}) \\ \head(LC) = p}} 
\left( 
\bigwedge_{\substack{l \in \body (LC)\\ l~\text{internal literal}}} l^{\FCM} \land 
\bigwedge_{\substack{u(RF) \in \body (LC)\\ RF \in \Facts (\textbf{P})}} u(RF)^{\FCM}
\right) 
\right\}_{p \in \mathfrak{I}(\mathfrak{P})}.
$$
Here, we find that the $u(RF)^{\FCM}$ are mutually independent Boolean random variables for every random fact~\mbox{$RF \in \Facts(\textbf{P})$}, each holding hold with probability~$\pi(RF)
$.  
\begin{example}
The FCM-semantics of the program $\textbf{P}$ in Example \ref{example - program} is given by 
\begin{align*}
& treatment^{\FCM} := u_1^{\FCM} && recovery^{\FCM} := u_2^{\FCM} \lor (treatment^{\FCM} \land u_3^{\FCM}),
\end{align*}
where $u_1^{\FCM}$, $u_2^{\FCM}$ and $u_3^{\FCM}$ are mutually independent Boolean random variables holding true with a probability of $0.5$, $0.5$ and $0.4$ respectively.
\end{example}

For the rest of this section, we fix a ProbLog program $\textbf{P}$ with an acyclic underlying logic program. Note that the FCM-semantics $\FCM (\textbf{P})$ yields a unique solution for every internal proposition~${p \in \mathfrak{I}(\mathfrak{P})}$ in terms of the mutually independent Boolean random variables $u(RF)^{\FCM}$. In this way, it defines a distribution on the $\mathfrak{P}$-structures $\mathcal{M}: \mathfrak{P} \rightarrow \{ True , False \}$ which coincides with the distribution semantics~\cite{DistributionSemantics} according to~Rückschloß and Weitkämper \cite{fcm-semantics}. Finally, we define the probability of a formula $\phi$ to hold as~${
\pi (\phi) := \sum_{\substack{\mathcal{M}~\mathfrak{P}\text{-structure} \\ \mathcal{M} \models \phi }} \pi (\mathcal{M})}.
$

However, the FCM-semantics does not only support queries about conditional and unconditional probabilities. It allows us to answer two more general causal query types, namely determining the effect of external interventions and counterfactuals \cite{fcm-semantics}. Assume for instance we want to \textbf{intervene} and set a subset $\textbf{X} \subseteq \mathfrak{I}(\mathfrak{P})$ of internal propositions to truth values specified by an assignment~$\textbf{x}$. In this case, we build a modified program~$\textbf{P}^{\Do (\textbf{x})}$ by erasing all clauses $LC \in \LP(\textbf{P})$ with head in $\textbf{X}$ and by adding a fact~${p \leftarrow}$ if $p \in \textbf{X}$ is set to true by $\textbf{x}$. If we now ask for the probability $\pi (\phi \vert \Do (\textbf{x}))$ of a formula $\phi$ to hold after setting $\textbf{X}$ to the values  $\textbf{x}$, we query the program $\textbf{P}^{\Do (\textbf{x})}$ for the probability of $\phi$.

\begin{example}
Assume we treat the patient in the program $\textbf{P}$ of Example \ref{example - program}. In this case, we obtain 
\begin{align*}
& 0.5 :: u_1 && 0.5 :: u_2 && 0.4 :: u_3 &&
& treatment \leftarrow  && recovery \leftarrow u_2 && recovery \leftarrow treatment, u_3
\end{align*}
for the modified program $\textbf{P}^{\Do(treatment)}$. This means we obtain a probability of $0.7$ for $recovery$ if we are the doctor and decide to treat our patient.  
\label{example - intervention}
\end{example}   

Further, we do not only want to either observe or intervene, we also want to know what the probability of an event would have been if we had intervened before observing some evidence. This is especially interesting in the \textbf{counterfactual} case where our evidence contradicts the given intervention.

\begin{example}
Consider the query in the introduction and observe that the evidence $\{ \neg treatment, recovery \}$ contradicts the intervention $\Do (treatment)$, i.e.~this is a counterfactual query.
\end{example} 

Hence, fix another subset of internal propositions $\textbf{E} \subseteq \mathfrak{I}(\mathfrak{P})$ and assume we observe the evidence that the propositions in $\textbf{E}$ take values according to the assignment $\textbf{e}$. We now ask for the probability~${\pi (\phi \vert \textbf{e}, \Do (\textbf{x}))}$ of the formula $\phi$ to hold if we had set the propositions in $\textbf{X}$ to the values specified by $\textbf{x}$ before observing our evidence $\textbf{e}$. To answer queries like that we proceed as Kiesel et al.~\cite{whatif}: 

First we generate two copies $\mathfrak{I}(\mathfrak{P})^{e/i}$ of the set of internal propositions -- one to handle the evidence and the other to handle the interventions. Further, we set $u^{e/i} := u$ for every external proposition~${u \in \mathfrak{E}(\mathfrak{P})}$. Note that this yields maps $\_^{e/i}$ of literals, clauses, programs etc. We define the \textbf{counterfactual semantics} of $\textbf{P}$ to be the ProbLog program $\textbf{P}^K$ which consists of the logic program ${\textbf{L}(\textbf{P})^e \cup \textbf{L}(\textbf{P})^i}$ and the random facts $\Facts(\textbf{P})$. Now we intervene in $\textbf{P}^K$ and set the proposition in $\textbf{X}^i$ to the truth values specified by $\textbf{x}$ to obtain the program~$\textbf{P}^{K,\Do(\textbf{x})}$. Finally, we query the program $\textbf{P}^{K,\Do(\textbf{x})}$ for the probability~$\pi(\phi^i \vert \textbf{E}^e = \textbf{e} )$ to obtain the desired result for $\pi (\phi \vert \textbf{e}, \Do (\textbf{x}))$.

\begin{example}
Assume we did not treat the patient in the program $\textbf{P}$ of Example~\ref{example - program} and he recovered. What is the probability $\pi (recovery \vert \neg treatment, recovery, \Do (treatment))$ that he would have recovered, if he had been treated?  To answer this question we query the program $\textbf{P}^{K, \Do(treatment)}$
\begin{align*}
& 0.5 :: u_1 && 0.5 :: u_2 && 0.4 :: u_3 \\
& treatment^i \leftarrow  && recovery^i \leftarrow u_2 && recovery^i \leftarrow treatment^i, u_3 \\
& treatment^e \leftarrow u_1  && recovery^e \leftarrow u_2 && recovery^e \leftarrow treatment^e, u_3
\end{align*}
for the probability $\pi (recovery^i \vert recovery^e, \neg treatment^e) = 1$.  
\label{example - intervention}
\end{example} 

This procedure automates Pearl's counterfactual reasoning \cite{Causality} and is implemented in the \textsc{WhatIf}-solver of Kiesel et al.~\cite{whatif}. 

In this contribution, we restrict ourselves to ProbLog programs, which can be represented with ProbLog clauses. A \textbf{ProbLog clause} $RC$ is an expression $~{\pi(RC) :: \effect(RC) \leftarrow causes(RC)}$, where ${\effect(RC) \in \mathfrak{I}(\mathfrak{P})}$ is an internal proposition called the \textbf{effect}, where $\causes(RC)$ is a finite set of internal literals called the \textbf{causes} and where $0 \leq \pi(RC) \leq 1$ is a number called the \textbf{probability} of $RC$. The ProbLog clause $RC$ is an abbreviation for the following pair of a random fact and a logical clause.
\begin{align*}
&RF(RC) := (\pi(RC) :: u(RC)) && LC(RC):= (h \leftarrow b_1,...,b_n,u(RC)),
\end{align*} 
where $u(RC) \in \mathfrak{E}(\mathfrak{P})$ is a distinct external literal. From now on, by abuse of language, a \textbf{ProbLog program} $\textbf{P}$ is a finite set of ProbLog clauses, i.e.~a ProbLog program consisting of the logic program ${\textbf{L}(\textbf{P}) := \{ LC(RC) \text{ : } RC \in \textbf{P} \}}$ and of the random facts $\Facts(\textbf{P}) := \{ RF(RC) \text{ : } RC \in \textbf{P} \}$.

\begin{example}
Observe that the program $\textbf{P}_1$ from the introduction is an abbreviation for the ProbLog program in Example \ref{example - program}. 
\end{example}

 The \textbf{class dependency graph} $\Graph (\textbf{P})$ of a ProbLog program $\textbf{P}$ is the directed graph on the internal propositions $\mathfrak{I}(\mathfrak{P})$ obtained by drawing an edge ${p_1 \rightarrow p_2}$ if and only if there exists a ProbLog clause~${RC \in \textbf{P}}$ with a cause $p_1$ or $\neg p_1$ and with effect $p_2$. We say that the program $\textbf{P}$ is \textbf{acyclic} if its class dependency graph $\Graph (\textbf{P})$ is a directed acyclic graph. Moreover, the program $\textbf{P}$ is \textbf{positive} if the causes of every ProbLog clause $RC \in \textbf{P}$ form a set of positive literals. 

\begin{example}
The class dependency graph of the programs $\textbf{P}_{1/2}$ in the introduction is given by the edge ${treatment \rightarrow recovery}$. Further, program $\textbf{P}_1$ is positive whereas $\textbf{P}_2$ is not.
\label{example - class dependency graph}
\end{example}

Next, we quickly recall the overview over the structure learning techniques for propositional ProbLog programs from Riguzzi \cite[§10]{PLP}. Given suitable data, all those algorithms search the space of programs defined by a language bias and background knowledge with heuristics relying on statistical tests. 

Here, the \textbf{language bias} defining the clause and therefore~the program space is defined by mode declarations of the form $modeb \left(*, q \right)$ and $modeh \left( * , p \right)$  as well as declarations of the form $determination (p,q)$ for propositions $p$ and $q$. A positive ProbLog clause $\pi :: \effect(RC) \leftarrow  \causes (RC)$ lies in the language defined by our bias if we declared $modeh(*,\effect(RC))$, $modeb \left(*, c \right)$ for all $c$ in $\causes (RC)$ and $determination \left( \effect (RC) , c \right)$ for all $c$ in $\causes (RC)$.

Moreover, we can express \textbf{background knowledge} in a logic program defining further propositions in terms of the given data. This is also the way how one can learn clauses with negation by adding a clause~${neg\_p \leftarrow \neg p}$ to the background knowledge for every proposition $p$. We call the pair of a language bias and a background knowledge the \textbf{setting} for a structure learning algorithm. 

\begin{example}
To learn the program $\textbf{P}_1$ of the introduction we nee to specify the setting $\mathfrak{S}_1$ which consists of the bias 
$ modeh(*,recovery)$, $modeb(*,treatment)$, $determinantion(recovery,treatment)$
and an empty background knowledge. If we want to consider the program $\textbf{P}_2$ as well, we additionally need the declarations $modeb(*,neg\_treatment)$ and $determinantion(recovery,neg\_treatment)$ together with the the clause ${neg\_treatment \leftarrow \neg treatment}$ in the background knowledge resulting in the setting~$\mathfrak{S}_2$.
\label{example - setting}
\end{example}

Note that via the $determination/2$ predicate the language bias essentially provides the class dependency graph, i.e.~the corresponding cause-effect relationships, of our program as prior knowledge to the structure learning algorithm. 

\section{Results}
Generally, in structure learning, from some prior knowledge encoded by a setting one wants to derive a program that describes a given set of data. In most of the cases, the data consists of observations. We additionally assume that our data consists~of samples drawn from the distribution induced by a hidden ProbLog program $\tilde{\textbf{P}}$ of interest and the prior knowledge consists of a language bias, i.e.~of the class dependency graph of $\tilde{\textbf{P}}$. Further, to decide how good a candidate program~$\textbf{P}$ represents our dataset we process statistical tests. However, statistical tests only measure how well the induced distribution of the program $\textbf{P}$ fits a given set of observations. They generally reveal no information about the causal mechanism generating our data. Hence, we cannot measure whether the causal mechanism represented by a candidate ProbLog program~$\textbf{P}$ coincides with the causal mechanism underlying our data, i.e.~with the causal mechanism described by $\tilde{\textbf{P}}$.
 
\begin{example}
Consider the programs $\textbf{P}_{1/2}$ of the introduction. While they both represent different causal models yielding to different counterfactual estimations, they yield the same distribution semantics and share the same class dependency graph.

Hence, if we take $\tilde{\textbf{P}} := \textbf{P}_{1/2}$ for the hidden program, we sample from the same distribution in both cases. That means that even with the correct language bias  a structure learning algorithm cannot determine which of the two programs actually generated the provided data unless it is given further knowledge. 
\label{example - statistical tests}
\end{example}

More drastically, Example \ref{example - statistical tests} illustrates that without further prior knowledge, even under the assumption of perfect learning, it is only possible to learn a program $\textbf{P}$ which is ensured to represent the correct distribution. In particular, we expect that a learned program $\textbf{P}$ does not necessarily answer counterfactual queries correctly. ´

In the following, we study the fragment of acyclic proper positive ProbLog programs in normal form. A ProbLog program $\textbf{P}$ is \textbf{proper in normal form} if every clause $RC \in \textbf{P}$ has a probability ${0 < \pi (RC) < 1}$, if any two distinct clauses $RC_{1/2} \in \textbf{P}$ have distinct causes ${\causes(RC_1) \neq \causes (RC_2)}$ or distinct effects ${\effect (RC_1) \neq \effect (RC_2)}$ and if every sink $s$ in the class dependency graph gives rise to a random fact~$\alpha :: s$.
The main result of this contribution now states that all programs lying in this fragment are uniquely determined by their class dependency graph and their underlying distribution. 

\begin{theorem}
Every acyclic proper positive ProbLog program in normal form $\textbf{P}$ can be reconstructed from its class dependency graph $\graph(\textbf{P})$ and the induced distribution $\pi$. 
\label{theorem - our fragment}
\end{theorem}

\begin{proof}
We proceed by induction on the number $n$ of nodes in the class dependency graph $\graph(\textbf{P})$.
\begin{enumerate}
\item[$n=1$:]
In this case, the program $\textbf{P}$ consists only of one clause $\pi :: p \leftarrow$. Hence, we set $\pi := \pi (p)$ and we are done.
\item[$n > 1$:]
Choose a sink $h \in \mathfrak{P}$ of $\graph(\textbf{P})$. Further, denote by $\textbf{P} \setminus h$ the program that results from $\textbf{P}$ if we erase all clauses with effect $h$. By maximality $h$ does not occur in the causes of any other clause, i.e.~$\textbf{P} \setminus h$ induces the same distribution on $\mathfrak{I}(\mathfrak{P}) \setminus \{ h \}$ as the program $\textbf{P}$ and it has the graph $\graph(\textbf{P}) \setminus h $ as its class dependency graph. Here, $\graph(\textbf{P}) \setminus  h $ denotes the graph that results from $\graph(\textbf{P})$ if we erase the node $h$ together with all edges pointing into it. Now, by the induction hypothesis we can reconstruct the program $\textbf{P} \setminus h$ from the given data.

Hence, we are left to reconstruct the clauses defining $h$ itself. Note that the parents $b \in \pa(h)$  of $h$ in $\graph(\textbf{P})$ are the only propositions that may occur in the body of a clause defining $h$. Further, note that each of these occurrences is positive. We consider the function
\begin{align*}
\Ind_h^{\textbf{P}} : & \mathcal{P}(\pa (h)) \rightarrow [0,1] 
& T \mapsto \pi (h \vert \{ t, \neg s \text{ : } t \in T,~s \in \pa(h) \setminus T \}),
\end{align*} 
where $\mathcal{P}(\_)$ denotes the power set operator. Recall from Pearl \cite[§3]{Causality} that observing the parents of $h$ is the same as intervening on them. Hence, we see that 
\begin{equation}
Ind_h^{\textbf{P}} (T) := 
\pi \left( \bigvee_{ \substack{RC \in \textbf{P} \\ \causes (RC) \subseteq T \\ \effect (RC) = h} } u(RC) \right)
= 
\sum_{ \substack{RC_1,...,RC_k \in \textbf{P} \\ k \in \mathbb{N},~ \causes (RC) \subseteq T
 \\ \effect (RC) = h}} (-1)^k \prod_{i=1}^k \pi(RC_i)
\label{equation - reconstruct parameters}
\end{equation} 
Now, it is easy to see that $T \subseteq \pa (\textbf{P})$ are the causes of a clause in $\textbf{P}$ if and only if ${\Ind_h^{\textbf{P}} (S) < \Ind_h^{\textbf{P}} (T)}$ for all $S \subseteq \pa (h)$ with $S \subsetneq T$. Further, we obtain the parameters of $\textbf{P}$ by recursion: We find that $\pi (RC) = \Ind_h^{\textbf{P}} (\body (RC))$ for every clause $RC \in \textbf{P}$ with a minimal body. Further, in the recursion step Equation (\ref{equation - reconstruct parameters}) yields a one-dimensional linear equation for the parameter of interest. 
\end{enumerate}
\end{proof}

In a forthcoming paper, we prove that the answers to all counterfactual queries uniquely determine a proper ProbLog program in normal form. Thus if we want to answer counterfactual queries based on a learned program $\textbf{P}$, we almost need to fully reconstruct the hidden program $\tilde{\textbf{P}}$.     

\begin{remark}
If we estimate the functions $Ind_h^{\textbf{P}}$ using relative frequencies, we obtain a structure learning algorithm that recovers an acyclic proper positive ProbLog program in normal form from a known causal structure and a sufficiently large set of samples. The resulting distributions (one for every counterfactual query) are guaranteed to converge in probability. Further, the complexity is exponential in the maximal number of parents of a node in the class dependency graph and linear in the size of the alphabet $\mathfrak{P}$.     
\end{remark}

Finally, assume we apply a structure learning algorithm \cite[§10]{PLP} with a language bias encoding the class dependency graph $\Graph(\tilde{\textbf{P}})$ to obtain a program~$\textbf{P}$. Let us further assume that we learned perfectly, i.e.~that the program~$\textbf{P}$ encodes the same distribution as $\tilde{\textbf{P}}$. If we now assume that both programs $\textbf{P}$ and~$\tilde{\textbf{P}}$ are acyclic proper positive ProbLog programs in normal form, Theorem \ref{theorem - our fragment} yields that $\textbf{P}$ and $\tilde{\textbf{P}}$ coincide, i.e.~$\textbf{P}$ expresses the full causal content of $\tilde{\textbf{P}}$. 

Since without background knowledge each currently available structure learning algorithm \cite[§10]{PLP} only searches for positive programs fitting a given dataset, the assumption to learn proper positive ProbLog programs in normal form is realized easily. Causally, the absence of background knowledge implies that we assume our data to be generated by a proper positive ProbLog program in normal form.    

\begin{corollary}
Assume we are given data sampled from a hidden proper positive ProbLog program in normal form $\tilde{\textbf{P}}$ and assume we are aware of the class dependency graph $\Graph(\tilde{\textbf{P}})$ of $\tilde{\textbf{P}}$. Every structure learning algorithm that is able to learn a  proper positive ProbLog program in normal form with the correct class dependency graph and the correct distribution reconstructs $\tilde{\textbf{P}}$ from the provided data. In particular, the result of such a structure learning algorithm supports counterfactual reasoning.~$\square$
\end{corollary}

\section{Conclusion}
In the introduction, we show that the distribution semantics does not uniquely determine counterfactual query outcomes for ProbLog programs, making it unfeasible to use the currently available structure learning algorithms for counterfactual reasoning. However, our main result reveals that proper positive ProbLog programs in normal form are actually uniquely determined by their distribution semantics and their class dependency graph. Hence, if applied without background knowledge and if we assume perfect learning, the currently available structure learning algorithms can recover these programs when provided with the correct language bias i.e.~they support counterfactual reasoning in this setting. 

In a forthcoming paper, we show that counterfactual reasoning uniquely determines a proper program in normal form, i.e.~it is not sufficient to learn programs under a coarser notion of equivalence. Determining the equivalence classes of ProbLog programs representing the same distributions and predicting the behaviour of the available structure learning algorithm for more general fragments of ProbLog are promising directions for future work extending this contribution.  
   
\nocite{*}
\bibliographystyle{eptcs}
\bibliography{literature.bib}
\end{document}